\documentclass[letterpaper, 10 pt, conference]{ieeeconf}  

\usepackage{amsmath}
\usepackage{amsfonts}
\usepackage{cite} 
\usepackage{subcaption}
\usepackage{graphicx}
\usepackage{siunitx}

\newtheorem{theorem}{Theorem}

\IEEEoverridecommandlockouts                              

\title{\LARGE \bf
Robust Adaptive Sliding-Mode Control for Damaged Fixed-Wing UAVs
}

\author{Mark Spiller, Lennart Kracke, and Johannes Autenrieb
\thanks{The authors are with the
	German Aerospace Center (DLR), Institute of Flight Systems, 38108, Braunschweig, Germany. {\tt\small mark.spiller@dlr.de, lennart.kracke@dlr.de,
	johannes.autenrieb@dlr.de}}%
}

\begin{document}

\maketitle
\thispagestyle{empty}
\pagestyle{empty}

\begin{abstract}

Many unmanned aerial vehicles (UAVs) can remain aerodynamically flyable after sustaining structural or control surface damage, yet insufficient robustness in conventional autopilots often leads to mission failure. This paper proposes a robust adaptive sliding mode controller (RASMC) for fixed-wing UAVs subject to aerodynamic coefficient perturbations and partial loss of control surface effectiveness. A damage-aware flight dynamics model is developed to systematically analyze the impact of such impairments on the closed-loop behavior. The RASMC is designed to ensure reliable tracking and stabilization, while a gain adaptation law maintains low control effort under nominal conditions and increases the gains as needed in the presence of aerodynamic damage. Lyapunov-based stability guarantees are derived, and assumptions on admissible uncertainty bounds are formulated to characterize the limits within which closed-loop stability and performance can be ensured. 
The proposed controller is implemented within an existing UAV autopilot framework, where outer-loop guidance and speed control modules provide reference commands to the RASMC for attitude stabilization. Simulations demonstrate that, despite significant damage, all closed-loop states remain stable with bounded tracking errors.
\end{abstract}

\section{Introduction}
\label{sec:introduction}
Structural damage or partial control surface degradation does not necessarily imply loss of aerodynamic controllability in fixed-wing unmanned aerial vehicles (UAVs). In many cases, the airframe remains physically capable of sustained flight despite significant impairment. Yet, conventional flight control systems are predominantly designed for nominal operating conditions and may lack the robustness required to accommodate severe aerodynamic parameter variations or actuator effectiveness loss. Consequently, vehicles that are aerodynamically survivable can nevertheless become unstable or exhibit substantial tracking degradation. Bridging this gap between physical survivability and control robustness represents a fundamental challenge in the development of damage-tolerant flight control systems.

Various approaches have been proposed to achieve stabilization and control of fixed-wing UAVs under conditions of damage \cite{castaneda2017extended,
	wang2019incremental,
	wang2019incremental2,
	bao2021design,
	zhao2014fault,
	he2020reconfigurable,
	kim2013flight,
	chowdhary2013guidance,
	du2025fault,
	nguyen2008flight,
	liu2010modeling,
	guo2011multivariable,
	nguyen2006dynamics,
	xiao2016modified,
	kim2019flight
	}. 
Most of these approaches fall into the categories of robust control methods, which aim to ensure stability under a wide range of uncertainties, adaptive control methods, which adjust to changing system dynamics in real-time, or robust adaptive approaches, which combine the benefits of both domains. 

Robust control methods are designed to provide stability despite uncertainties such as actuator faults, aerodynamic disturbances, or structural damage. In \cite{castaneda2017extended}, sliding-mode control (SMC) in combination with extended state observer is applied for robust attitude stabilization of fixed-wing UAVs. However, the approach does not explicitly account for aerodynamic uncertainties and is limited in the sense that it cannot handle degradation of control-surface effectiveness. Incremental dynamics-based SMC approaches are proposed by \cite{wang2019incremental,wang2019incremental2} with consideration of actuator faults, including partial loss of control effectiveness, as well as structural damage. The stability analysis relies on lumped uncertainty bounds for the control effectiveness matrix, which is overly conservative and the proofs apply incremental control inputs within a time-continuous Lyapunov framework, which introduces technical challenges. In addition, the approaches use static gains and can not adapt to the amount of damage being present. The work of \cite{bao2021design} employs a SMC-based backstepping approach for robust fixed-wing UAV control. It is capable of handling matched disturbances but assumes perfect knowledge of the control-surface effectiveness matrix. The authors in \cite{zhao2014fault} design a SMC approach based on a linearized aircraft model that incorporates modeled damage. Although robustness is achieved, the control law requires prior knowledge of the structure of the damage-related input matrix, in contrast to most existing approaches, which assume only bounds on the uncertainties. In \cite{he2020reconfigurable}, robust control with respect to matched uncertainties is considered. Nonlinear dynamic inversion (NDI) in combination with a disturbance observer is applied for the rejection of the uncertainties.

Adaptive control methods aim to adjust the control laws in real-time based on changes in the system dynamics, such as those caused by damage \cite{ioannou1996robust}. These methods are particularly useful when the system is subject to unknown or varying conditions. In \cite{kim2013flight}, a neural-network-based adaptive scheme is proposed to compensate for matched uncertainties in a fixed-wing UAV experiencing partial wing loss. The method is based on NDI and therefore assumes perfect knowledge of the control effectiveness matrix. The approach of \cite{chowdhary2013guidance} employs model reference adaptive control (MRAC) to compensate for matched uncertainties in the UAV inner-loop dynamics and was successfully validated in flight tests with a $\SI{25}{\%}$ right-wing loss. The controller relies on inversion of a nominal model and assumes purely additive uncertainties. In \cite{du2025fault}, a control strategy for fixed-wing aircraft with asymmetric wing damage is proposed. Based on linearized dynamics a trimming algorithm is combined with an adaptive control allocation for an attitude PID controller. However, the PID gains themselves are not adapted in response to the damage. The authors in \cite{nguyen2008flight} address robustness against shifts in the center of gravity and perturbations in aerodynamic coefficients for fixed-wing aircraft. Based on a linearized model, the approach combines NDI with a PI controller and an adaptive neural-network term. The adaptive component is designed to compensate only for matched uncertainties, while NDI assumes accurate knowledge of the control effectiveness matrix. The work of \cite{liu2010modeling,guo2011multivariable} introduces a damage model based on perturbations and linearizes the aircraft dynamics around a nominal operating condition. A MRAC scheme is developed to compensate for changes in the system dynamics but relies on a linear approximation of the true nonlinear system.

While previous works have addressed robust or adaptive control for damaged UAVs, most assume perfect knowledge of the control effectiveness matrix. In addition, MRAC approaches rely on linearized dynamics, whereas most existing SMC methods can handle the full nonlinear dynamics but are not adaptive, requiring unnecessarily high gains.

In this paper, we present a robust adaptive sliding mode controller (RASMC) for nonlinear UAV dynamics under aerodynamic perturbations and degraded control surface effectiveness. We formulate an adaptation law for the controller gains, allowing them to remain low under nominal conditions and increase in the presence of damage. Lyapunov-based stability proofs are provided that guarantee the convergence of tracking error, offering formal assurances of stability even in presence of model mismatch. Additionally, we formulate conditions for the bounds of uncertainties under which closed-loop stability can be achieved, providing a clear framework for understanding the limits within which the proposed method remains stable. Finally, we demonstrate the effectiveness of the proposed method by embedding the RASMC into a full autopilot and simulating the UAV dynamics under damage.

\section{Modeling of Damaged Fixed-Wing UAV}
\label{sec:modeling_UAV}
The considered translational and rotational dynamics of the fixed-wing UAV are expressed in the body-fixed frame according to
\begin{subequations}\label{num:rigid_body_dynamics}
	\begin{align}
		\dot V
		&
		=
		-
		\omega \times V
		+
		\frac{F}{m},\\
		\dot \Theta
		&=
		\Psi\omega,\\
		\dot \omega &= 
		-J^{-1}(\omega \times J\omega)
		+
		J^{-1}M,
	\end{align}
\end{subequations}
with the states being: the body-fixed velocities
\( 
V=
\begin{bmatrix}
	u & v & w
\end{bmatrix}^T
\in \mathbb{R}^3
\), 
the rates around the axis of the body-fixed frame 
\( 
\omega=
\begin{bmatrix}
	p & q & r
\end{bmatrix}^T
\in \mathbb{R}^3
\) 
and the Euler angles
\( 
\Theta=
\begin{bmatrix}
	\phi & \theta & \psi
\end{bmatrix}^T
\in \mathbb{R}^3
\). The inertia tensor is given by $J\in\mathbb{R}^{3 \times 3}$ and the total mass of the rigid body is $m\in \mathbb{R}_{>0}$. The inputs that drive the dynamics are the total forces $F\in\mathbb{R}^{3}$ and moments $M\in\mathbb{R}^{3}$ expressed in the body-fixed frame. The Euler angle rate transformation matrix $\Psi\in\mathbb{R}^{3 \times 3}$ is
\begin{align}
	\Psi
	=
	\begin{bmatrix}
		1 
		& \sin(\phi)\tan(\theta)
		& \cos(\phi)\tan(\theta)\\
		0 
		&
		\cos(\phi)
		&
		-\sin(\phi)\\
		0
		&
		\sin(\phi)/\cos(\theta)
		&
		\cos(\phi)/\cos(\theta)
	\end{bmatrix}.\notag
\end{align}
Furthermore, we define the following quantities. The aerodynamic velocities are
\begin{align}
	u_{\mathrm{air}}=u-u_w, \quad
	v_{\mathrm{air}}=v-v_w, \quad
	w_{\mathrm{air}}=w-w_w, \notag
\end{align}
where $u_w$, $v_w$ and $w_w$ are the wind speed components in the body-fixed frame. The dynamic pressure is calculated as
\(
Q = 
0.5 \rho_{\mathrm{air}} v_{\mathrm{TAS}}^2 \notag
\), where $\rho_{\mathrm{air}}\in\mathbb{R}_{>0}$ is the air density and
\begin{align}
	v_{\mathrm{TAS}}
	=
	\sqrt{u_{\mathrm{air}}^2
		+
		v_{\mathrm{air}}^2
		+
		w_{\mathrm{air}}^2}
		\notag
\end{align}
is the true airspeed. The matrix 
\begin{align}
	T_{fa} =
	\begin{bmatrix}
		\cos(\alpha)\cos(\beta)
		& -\cos(\alpha)\sin(\beta)
		& -\sin(\alpha)\\
		\sin(\beta) & \cos(\beta) & 0\\
		\sin(\alpha)\cos(\beta) & -\sin(\alpha)\sin(\beta) & \cos(\alpha)
	\end{bmatrix} 
	\notag
\end{align}
is the rotational matrix from the aerodynamic frame to the body-fixed frame, where 
\(
\alpha = \text{atan2}(w_{\mathrm{air}},u_{\mathrm{air}})
\)
is the angle of attack and
\(
\beta = \text{arcsin}(\frac{v_{\mathrm{air}}}{v_{\mathrm{TAS}}})
\)
is the sideslip angle. We denote $S_{\mathrm{ref}}, C_{\mathrm{ref}}, B_{\mathrm{ref}}\in\mathbb{R}_{>0}$ as the reference wing area, the reference chord length and the reference span, respectively, and we define normalized rates as 
\begin{align}
	\bar p = \frac{B_{\mathrm{ref}}}{2 v_{\mathrm{TAS}}}
	p,
	\quad 
	\bar q = 
	\frac{C_{\mathrm{ref}}}{2 v_{\mathrm{TAS}}}q,
	\quad  
	\bar r = 
	\frac{B_{\mathrm{ref}}}{2 v_{\mathrm{TAS}}}r
	\notag.
\end{align}
Additionally, we denote
$\xi, \eta, \zeta\in\mathbb{R}$ as the aileron, elevator, and rudder deflections, respectively, and we define  
\(
\delta
=
\begin{bmatrix}
	\xi & \eta & \zeta
\end{bmatrix}^T
\) as the vector of control surface deflections. We denote the thrust by $T\in\mathbb{R}_{>0}$.

Aerodynamic damage is typically modeled by modifying the aircraft’s aerodynamic coefficients. Each affected coefficient may be scaled and augmented by an additive bias term to represent asymmetric structural damage. This approach captures both the reduction or alteration of aerodynamic derivatives and the emergence of destabilizing forces and moments \cite{nguyen2006dynamics,shah2008aerodynamic}. In the calculation of the aerodynamic coefficients, we introduce damage factors, which modify the contribution of the corresponding aerodynamic derivatives to account for structural damage. Each derivative is scaled by a multiplicative factor 
\(
\lambda_{\star,i}=(1-HD_{\star,i})
\), where 
\begin{align}
	H(t)
	 =
	\begin{cases}
		0\colon & t<t_d, \\
		1\colon & t\ge t_d
	\end{cases}
	\notag
\end{align}
is a Heaviside function representing the onset of damage at the time instant $t_d$, and $D_{\star,i}\in [0,1]$ is a scaling variable to account for the amount of damage being present from thereon. In addition, an additive asymmetric term 
\(
\lambda_{\star}^{\mathrm{asym}}= 
H\Delta C_{\star}^{\mathrm{asym}}
\), where 
\(
\Delta C_{\star}^{\mathrm{asym}}\in \mathbb{R}
\), is included to model forces or moments that arise due to asymmetry. This formulation captures both the alteration of the nominal aerodynamic derivatives and the introduction of destabilizing moments resulting from partial or asymmetric damage.

The total forces in the body-fixed frame are described by
\begin{align}
	F = F_{\mathrm{aero}} + F_{\mathrm{prop}} + F_{\mathrm{grav}},
	\label{num:total_forces_gen_equation}
\end{align}
where $F_{\mathrm{aero}}$ are the aerodynamic forces, $F_{\mathrm{prop}}$ are the propulsive forces, and $F_{\mathrm{grav}}$ are the gravitational forces. The propulsive and gravitational forces are given by
\begin{align}
	F_{\mathrm{prop}} 
	= 
	\begin{bmatrix}
		T \\  
		0 \\ 
		0
	\end{bmatrix}, \qquad
	F_{\mathrm{grav}}
	=
	\begin{bmatrix}
		-\sin(\theta)mg\\
		\sin(\phi)\cos(\theta)mg\\
		\cos(\phi)\cos(\theta)mg
	\end{bmatrix},
	\label{num:prop_and_grav_forces}
\end{align}
where we assume a nominal value of $g=9.80665$ and treat deviations of the gravitational constant as uncertainty. The aerodynamic forces are modeled as 
\begin{align}
	F_{\mathrm{aero}} = 
	T_{fa} Q S_{\mathrm{ref}} 
	\begin{bmatrix}
		-C_D\\
		C_Y\\
		-C_L
	\end{bmatrix}
	\label{num:aero_forces_gen_equation}
\end{align}
and the aerodynamic coefficients for lift, side-force and drag are specified by
\begin{subequations}\label{num:forces_aero_coeff_under_damage}
	\begin{align}
		C_L 
		&= 
		\bar C_L
		+
		\lambda_{L,\xi}C_{L,\xi} \xi
		+
		\lambda_{L,\eta}C_{L,\eta} \eta
		+
		\lambda_{L,\zeta}C_{L,\zeta} \zeta\\
		C_Y
		&=
		\bar C_Y
		+
		\lambda_{Y,\xi}C_{Y,\xi} \xi
		+
		\lambda_{Y,\eta}C_{Y,\eta} \eta
		+
		\lambda_{Y,\zeta}C_{Y,\zeta} \zeta\\
		C_D
		&=
		\bar C_D
	\end{align}
\end{subequations}
and
\begin{align}
	\bar C_L 
	&= 
	\lambda_{L,0}C_{L,0} 
	+ 
	\lambda_{L,\alpha}C_{L,\alpha}\alpha
	+
	\lambda_{L,\beta}C_{L,\beta}\beta
	+
	\lambda_{L,p}C_{L,p}\bar p
	\notag\\
	& \qquad
	+
	\lambda_{L,q}C_{L,q}\bar q
	+
	\lambda_{L,r}C_{L,r}\bar r 
	+
	\lambda_{L}^{\mathrm{asym}}
	\notag\\
	\bar C_Y
	&=
	\lambda_{Y,0}C_{Y,0}
	+
	\lambda_{Y,\alpha}C_{Y,\alpha} \alpha
	+
	\lambda_{Y,\beta}C_{Y,\beta} \beta
	+
	\lambda_{Y,p}C_{Y,p} \bar p
	\notag\\
	& \qquad
	+
	\lambda_{Y,q}C_{Y,q} \bar q
	+
	\lambda_{Y,r}C_{Y,r} \bar r 
	+
	\lambda_{Y}^{\mathrm{asym}}
	\notag\\
	\bar C_D
	&=
	C_{D,0}
	+
	\frac{1}{\pi e_0AR}C_L^2 
	+
	\lambda_{D}^{\mathrm{asym}}
	\notag.
\end{align}
In (\ref{num:forces_aero_coeff_under_damage}), the multiplicative factors $\lambda_{\star,i}$ and the bias terms $\lambda_{\star}^{\mathrm{asym}}$ account for the modeled damage effects. Further, Oswald efficiency number and the aspect ratio are denoted by $e_0$ and $AR$.
Substituting (\ref{num:prop_and_grav_forces}), (\ref{num:aero_forces_gen_equation}), (\ref{num:forces_aero_coeff_under_damage}) into (\ref{num:total_forces_gen_equation}) gives the total forces as 
\begin{align}
	F 
	&= 
	F_a
	+ 
	F_{\delta}\delta
	+ 
	F_{T}T
	\label{num:aero_forces_explict_expression}
\end{align}
with
\begin{align}
	F_a
	&= 
	\begin{bmatrix}
		-\sin(\theta)mg\\
		\sin(\phi)\cos(\theta)mg\\
		\cos(\phi)\cos(\theta)mg
	\end{bmatrix}
	+
	T_{fa} Q S_{\mathrm{ref}} 
	\begin{bmatrix}
		-\bar C_D\\
		\bar C_Y\\
		-\bar C_L
	\end{bmatrix}
	\notag\\
	F_{\delta}
	&= 
	T_{fa} Q S_{\mathrm{ref}} \cdot\notag\\
	& \qquad
	\begin{bmatrix}
		0 & 0 & 0\\
		\lambda_{Y,\xi}C_{Y,\xi} 
		&
		\lambda_{Y,\eta}C_{Y,\eta} 
		&
		\lambda_{Y,\zeta}C_{Y,\zeta}\\
		-
		\lambda_{L,\xi}C_{L,\xi} 
		&
		-
		\lambda_{L,\eta}C_{L,\eta} 
		&
		-
		\lambda_{L,\zeta}C_{L,\zeta} 
	\end{bmatrix}
	\notag
\end{align}
and
\(
F_{T}
=
\begin{bmatrix}
	1 & 0 &	0
\end{bmatrix}^T
\).

We assume the propulsive force to be aligned with the aircraft's centerline, meaning that only the aerodynamic moments need to be considered in the calculation of $M$. The aerodynamic moments are
\begin{align}
	M = 
	T_{fa} Q S_{\mathrm{ref}} 
	\begin{bmatrix}
		C_l B_{\mathrm{ref}}\\
		C_m C_{\mathrm{ref}}\\
		C_n B_{\mathrm{ref}}
	\end{bmatrix}	
	\label{num:aero_moments_gen_equation}
\end{align}
with aerodynamic coefficients for roll, pitch, and yaw being defined by
\begin{subequations}\label{num:moments_aero_coeff_under_damage}
\begin{align}
	C_l &=
	\bar C_l 
	+ 
	\lambda_{l,\xi}
	C_{l,\xi} \xi 
	+ 
	\lambda_{l,\eta}
	C_{l,\eta}  \eta 
	+ 
	\lambda_{l,\zeta} 
	C_{l,\zeta}\zeta \\
	C_m &=
	\bar C_m 
	+ 
	\lambda_{m,\xi}C_{m,\xi} \xi 
	+ 
	\lambda_{m,\eta}C_{m,\eta} \eta 
	+ 
	\lambda_{m,\zeta}
	C_{m,\zeta} \zeta 
	\\
	C_{n} 
	&= \bar C_{n}
	+ 
	\lambda_{n,\xi}
	C_{n,\xi} \xi 
	+ 
	\lambda_{n,\eta}
	C_{n,\eta} \eta 
	+ 
	\lambda_{n,\zeta}
	C_{n,\zeta} \zeta 
\end{align}
\end{subequations}
and
\begin{align}
	\bar C_l &= 	
	\lambda_{l,0}C_{l,0} 
	+ 
	\lambda_{l,\alpha}C_{l,\alpha} \alpha   
	+ 
	\lambda_{l,\beta}C_{l,\beta} \beta 
	+
	\lambda_{l,p}C_{l,p} \bar p 
	\notag\\
	& \quad  
	+ 
	\lambda_{l,q}C_{l,q}  \bar q 
	+
	\lambda_{l,r}C_{l,r} \bar r 
	+
	\lambda_{l}^{\mathrm{asym}}
	\notag\\
	\bar C_m &= 	
	\lambda_{m,0}C_{m,0} 
	+ 
	\lambda_{m,\alpha}C_{m,\alpha} \alpha 
	+ 
	\lambda_{m,\beta}
	C_{m,\beta} \beta 
	\notag\\
	& \quad 
	+
	\lambda_{m,p}
	C_{m,p} \bar p 
	+
	\lambda_{m,q}
	C_{m,q} \bar q 
	+
	\lambda_{m,r} 
	C_{m,r} \bar r  
	+
	\lambda_{m}^{\mathrm{asym}}
	\notag\\
	\bar C_{n}  &=	
	\lambda_{n,0} C_{n,0} 
	+ 
	\lambda_{n,\alpha}C_{n,\alpha} \alpha 
	+ 
	\lambda_{n,\beta}C_{n,\beta}\beta
	+
	\lambda_{n,p}C_{n,p} \bar p 
			\notag\\
	& \quad 
	+ 
	\lambda_{n,q}C_{n,q} \bar q 
	+
	\lambda_{n,r}C_{n,r} \bar r 			
	+
	\lambda_{n}^{\mathrm{asym}}
	\notag.
\end{align}
Again, the multiplicative factors $\lambda_{\star,i}$ and the bias terms $\lambda_{\star}^{\mathrm{asym}}$ account for the modeled damage effects. Substituting the expressions of the aerodynamic coefficients (\ref{num:moments_aero_coeff_under_damage}) into (\ref{num:aero_moments_gen_equation}) yields
\begin{align}
	M = 
	M_a
	+
	M_{\delta}
	\delta
	\label{num:aero_moments_explict_expression},	
\end{align}
where
\begin{align}
	M_a 
	&= 
	Q S_{\mathrm{ref}} T_{fa}
	\begin{bmatrix}
		B_{\mathrm{ref}}\bar C_l\\
		C_{\mathrm{ref}}\bar C_m\\
		B_{\mathrm{ref}}\bar C_{n}
	\end{bmatrix},	\notag\\
	M_{\delta} 
	&= Q S_{\mathrm{ref}} T_{fa}\cdot\notag\\
		& 
		\begin{bmatrix}
			B_{\mathrm{ref}}
			\lambda_{l,\xi}C_{l,\xi} & B_{\mathrm{ref}}
			\lambda_{l,\eta}C_{l,\eta}  & 
			B_{\mathrm{ref}}
			\lambda_{l,\zeta}C_{l,\zeta} \\
			C_{\mathrm{ref}}
			\lambda_{m,\xi}C_{m,\xi} & 
			C_{\mathrm{ref}}
			\lambda_{m,\eta}C_{m,\eta} & 
			C_{\mathrm{ref}}
			\lambda_{m,\zeta}C_{m,\zeta}\\
			B_{\mathrm{ref}}
			\lambda_{n,\xi}C_{n,\xi} & 
			B_{\mathrm{ref}}
			\lambda_{n,\eta}C_{n,\eta} & 
			B_{\mathrm{ref}}
			\lambda_{n,\zeta}C_{n,\zeta} 
		\end{bmatrix}\notag.
\end{align}

Substituting the forces (\ref{num:aero_forces_explict_expression}) and moments (\ref{num:aero_moments_explict_expression}) into the Newton-Euler equations (\ref{num:rigid_body_dynamics}) results in the UAV dynamics
\begin{subequations}\label{num:rigid_body_dynamics_explicit}
		\begin{align}
		\dot V
		&
		=
		-
		\omega \times V
		+
		\frac{F_a
			+ F_{\delta}\delta
			+ F_{T}T}{m},\\
		\dot \Theta
		&=
		\Psi\omega,
		\label{num:Euler_angle_dynamics}\\
		\dot \omega &= 
		-J^{-1}(\omega \times J\omega)
		+
		J^{-1}(	M_a
		+
		M_{\delta}
		\delta)
		\label{num:rate_dynamics}
	\end{align}
\end{subequations}
with modeled damage. 

For the subsequent controller design, we decompose the aerodynamic moment vector and control effectiveness matrix as
\begin{align}
	M_a
	&=
	M_{a,0}
	+
	\Delta
	M_{a}
	\label{num:decomp_Ma},\\
	M_\delta
	&=
	(I+\Delta
	M_{\delta})
	M_{\delta,0}.
	\label{num:decomp_Mdelta}
\end{align}
Here, $M_{a,0}\in\mathbb{R}^{3}$ and $M_{\delta,0}\in\mathbb{R}^{3\times 3}$ denote the known nominal aerodynamic moments and nominal control effectiveness matrix, respectively. The terms $\Delta M_a\in\mathbb{R}^{3}$ and $\Delta M_{\delta}\in\mathbb{R}^{3\times 3}$ represent unknown perturbations capturing damage-induced variations and modeling uncertainties. The decomposition in \eqref{num:decomp_Ma} and \eqref{num:decomp_Mdelta} enables the incorporation of nominal system knowledge into the controller design, effectively reducing the required control gains to those necessary for uncertainty compensation.

For the nominal aerodynamic properties of the UAV, we assume $M_{\delta,0}$ to be nonsingular, ensuring attitude controllability. This is a standard assumption, since the nominal aircraft design is intended to be fully controllable.

Furthermore, we restrict the Euler angles to the nonsingular region such that $\Psi$ remains invertible. This is ensured in practice by appropriate autopilot design, including reference shaping and limiting of commanded Euler angles to avoid operating near singular configurations.

\section{Robust Attitude Controller Design}
\label{sec:rob_controller}
Building upon the challenges outlined in Section \ref{sec:introduction}, we address aerodynamic damage–induced uncertainties through the development of a robust sliding mode attitude controller with adaptive gain tuning. The proposed framework provides Lyapunov-based stability guarantees and ensures tracking error convergence in the presence of parametric variations and actuator effectiveness loss.

The control design is developed in two stages. In Theorem~\ref{num:Theorem_SMC}, a sliding-mode attitude controller with fixed gains is derived and Lyapunov-based analysis is used to establish tracking error stability. Building on this, Theorem~\ref{num:Theorem_SMC_adaptive} introduces adaptive gain tuning, which also guarantees Lyapunov-based stability and reduces the conservatism of the static gain approach. The proposed robust adaptive sliding mode control (RASMC) method is defined based on the result of Theorem~\ref{num:Theorem_SMC_adaptive}.

\begin{theorem}[Static gain SMC]
	\label{num:Theorem_SMC}
	Consider the control law
	\begin{align}
		\delta
		&=
		M^{-1}_{\delta,0} J \Psi^{-1}
		(\Pi
		+
		\mu
		)
		\label{num:main_control_law}
	\end{align}
	with
	\begin{align}
		\Pi
		&=
		-
		\Psi J^{-1}M_{a,0}
		+
		\Psi J^{-1}(\omega \times J\omega)
		-
		\dot\Psi\omega,
		\label{num:aux_control_law_pi}\\
		\mu
		&=
		\ddot\Theta_r
		-
		\Lambda \dot {\tilde \Theta}
		-
		K
		\mathrm{sgn}(s)	
		\label{num:aux_control_law_mu}
	\end{align}
	where
	\(
	\tilde \Theta=\Theta-\Theta_r
	\) is the control tracking error,  
	\(
	\Theta_r
	=
	\begin{bmatrix}
		\phi_r & \theta_r &	\psi_r
	\end{bmatrix}^T
	\)
	are the reference attitude angles,
	$\Lambda\in\mathbb{R}^{3\times3}$
	is chosen so that matrix 
	\(
	\bar \Lambda=-\Lambda
	\) is Hurwitz,	
	\(
	K=\mathrm{diag}(k)\in\mathbb{R}^{3\times3}
	\) 
	is a diagonal matrix of the controller gains
	\( 
	k
	=
	\begin{bmatrix}
		k_1 & k_2 & k_3
	\end{bmatrix}^T
	\)
	and 
	\begin{align}
		s
		&=
		\begin{bmatrix}
			s_1 & s_2 & s_3
		\end{bmatrix}^T
		=
		\dot {\tilde \Theta}
		+
		\Lambda \tilde \Theta
		\label{num:def_sliding_surface}
	\end{align}
	is the sliding variable.
	Let $B_{ij}\in\mathbb{R}$ be bounds of 
	\(
	\forall t\ge0\colon
	|\Xi_{ij}(t)|
	\le
	B_{ij}
	\),
	with matrix elements $[\Xi_{ij}]=\Xi\in\mathbb{R}^{3\times3}$ being defined by
	\begin{align}
		\Xi
		&= 
		\Psi J^{-1} 
		\Delta
		M_{\delta}
		J\Psi^{-1}
		\label{mum:def_Sigma}
	\end{align}
	and let $a_i\in\mathbb{R}$ be bounds of
	\(
	\forall t\ge0\colon
	|\iota_i(t)|
	\le
	a_i
	\),
	with elements $\iota_i$ being defined by the vector
	\begin{align}
		\begin{bmatrix}
			\iota_1 & \iota_2 & \iota_3
		\end{bmatrix}^T
		=
		\Psi J^{-1}
		\Delta
		M_{a}
		+ 
		\Xi
		(
		\Pi
		+
		\ddot\Theta_r
		-
		\Lambda \dot {\tilde \Theta})
		\label{mum:def_iota}.
	\end{align}
	Denote 
	\begin{align}
		D
		=
		\begin{bmatrix}
			0 & 
			\frac{
				B_{12}
			}
			{1-B_{11}} &
			\frac{B_{13}}
			{1-B_{11}}\\
			\frac{B_{21}}
			{1-B_{22}} & 
			0 &
			\frac{B_{23}}
			{1-B_{22}}\\
			\frac{B_{31}}
			{1-B_{33}} & 
			\frac{B_{32}}
			{1-B_{33}}
			&
			0		
		\end{bmatrix},
		\quad
		z
		=
		\begin{bmatrix}
			\frac{a_1
				+\epsilon_1}{1
				-
				B_{11}}\\
			\frac{a_2
				+\epsilon_2}{1
				-
				B_{22}}	\\
			\frac{a_3
				+\epsilon_3}{1
				-
				B_{33}
			}		
		\end{bmatrix}
		\label{mum:def_D_and_z},
	\end{align}
	where $\epsilon_i>0$ is a tuning parameter. If the spectral radius $\varrho(D)<1$ and $\forall i \colon B_{ii}<1$
	then choosing the controller gains $k$ as
		\begin{align}
		(I_{3 \times 3}-D)	
		k
		=z
		\label{mum:k_eq_sys_static}
	\end{align}
	guarantees the tracking error to converge to zero, i.\,e. $\tilde \Theta\to0$ for $t\to\infty$.
\end{theorem}
\begin{proof}
	Consider the Lyapunov function candidate
	\begin{align}
		V_1
		&=
		\frac{1}{2}s^Ts.
		\label{num:def_LyapunovFct1}
	\end{align}
	Substituting 
	\(
	\dot s
	=
	\ddot {\tilde \Theta}
	+
	\Lambda \dot{\tilde \Theta}
	\) 
	from (\ref{num:def_sliding_surface}) into 
	\(
	\dot V_1
	=
	s^T\dot s
	\) 
	yields
	\begin{align}
		\dot V_1
		&=
		s^T
		(
		\ddot \Theta
		- 
		\ddot\Theta_r
		+
		\Lambda \dot{\tilde \Theta}
		)
		\label{mum:CLF_derivative1}.
	\end{align}
	Differentiating (\ref{num:Euler_angle_dynamics}), leads to 
	\(
	\ddot \Theta
	=
	\dot\Psi\omega
	+
	\Psi\dot\omega
	\) 
	and substituting (\ref{num:rate_dynamics}), (\ref{num:decomp_Ma}), (\ref{num:decomp_Mdelta})
	gives
	\begin{align}
	\ddot \Theta
	&=
	\dot\Psi\omega
	-
	\Psi J^{-1}(\omega \times J\omega)
	+
	\Psi J^{-1}M_{a,0}
	+
	\Psi J^{-1}
	\Delta
	M_{a}\notag\\
	& \quad
	+
	\Psi J^{-1} M_{\delta,0}\delta
	+
	\Psi J^{-1} \Delta M_{\delta}M_{\delta,0}\delta
	\label{num:dyn_Theta_sub1}
	\end{align}
	Substituting (\ref{num:main_control_law}) into (\ref{num:dyn_Theta_sub1}) leads to
	\begin{align}
		\ddot \Theta
		&=
		\dot\Psi\omega
		-
		\Psi J^{-1}(\omega \times J\omega)
		+
		\Psi J^{-1}M_{a,0}
		+
		\Psi J^{-1}
		\Delta
		M_{a}\notag\\
		& \quad
		+
		\Pi 
		+ 
		\mu
		+
		\Psi J^{-1} \Delta M_{\delta}J\Psi^{-1} 
		(\Pi + \mu)
		\label{num:dyn_Theta_sub1b}
	\end{align}
	and rearranging (\ref{num:dyn_Theta_sub1b}) with substitution of (\ref{num:aux_control_law_pi}) implies
	\begin{align}
		\ddot \Theta
		&=
		\Psi J^{-1}
		\Delta
		M_{a}
		+
		\Psi J^{-1} \Delta M_{\delta}J\Psi^{-1} 
		\Pi \notag\\
		& \qquad
		+
		(
		I_{3 \times 3}
		+ 
		\Psi J^{-1} \Delta M_{\delta}J\Psi^{-1} 
		)
		\mu.
		\label{num:dyn_Theta_sub2}
	\end{align}
	Based on the definition of $\Xi$ according to (\ref{mum:def_Sigma}), the expression (\ref{num:dyn_Theta_sub2}) is rewritten as
	\begin{align}
		\ddot \Theta
		&=
		\Psi J^{-1}
		\Delta
		M_{a}
		+
		\Xi
		\Pi
		+
		(
		I_{3 \times 3}
		+
		\Xi
		)
		\mu
		\label{num:dyn_Theta_sub3}.
	\end{align}
	Substituting (\ref{num:dyn_Theta_sub3}) into (\ref{mum:CLF_derivative1}) 
	results in
	\begin{align}
		\dot V_1
		&=
		s^T
		(
		\Psi J^{-1}
		\Delta
		M_{a}
		+
		\Xi
		\Pi
		+
		(
		I_{3 \times 3}
		+
		\Xi
		)
		\mu
		- 
		\ddot\Theta_r
		+
		\Lambda \dot{\tilde \Theta}
		)
		\label{mum:CLF_derivative2a}.
	\end{align}
	and substituting (\ref{num:aux_control_law_mu}) into (\ref{mum:CLF_derivative2a}) yields
	\begin{align}
		\dot V_1
		&=
		s^T
		\big(
		\Psi J^{-1}
		\Delta
		M_{a}
		+
		\Xi
		(
		\Pi
		+
		\ddot\Theta_r
		-
		\Lambda \dot {\tilde \Theta}
		)\notag\\
		& \qquad
		-
		(
		I_{3 \times 3}
		+
		\Xi
		)
		K
		\mathrm{sgn}(s)	
		\big)
		\label{mum:CLF_derivative2}.
	\end{align}
	Using the definition of $\iota_i$ from (\ref{mum:def_iota}) the expression (\ref{mum:CLF_derivative2}) can be written as
	\begin{align}
		\dot V_1
		&=
		\sum_{i=1}^{3}
		\Bigg
		(
		-
		|s_i|
		(1
		+
		\Xi_{ii}
		)k_i
		-
		s_i
		\sum_{\substack{j=1 \\ j \neq i}}^{3}
		\Xi_{ij}k_j
		\text{sgn}(s_j)
		+
		s_i\iota_i
		\Bigg
		)
		,\notag
	\end{align}
	which leads to
	\begin{align}
		\dot V_1
		\le
		\sum_{i=1}^{3}
		|s_i|
		\Bigg
		(
		-
		(1
		+
		\Xi_{ii}
		)
		k_i
		+
		\sum_{\substack{j=1 \\ j \neq i}}^{3}
		|\Xi_{ij}|k_j
		+
		|\iota_i|
		\Bigg
		)
		\label{mum:CLF_derivative3},
	\end{align}
	where it is assumed that the gains $k_1, k_2, k_3$ are all non-negative.
	To achieve $\forall s\not=0\colon\dot V_1<0$, we want $\forall i,j \in\{1,2,3\}\colon$
	\begin{align}
		-
		(1
		+
		\Xi_{ii}
		)
		k_i
		+
		\sum_{j\not=i}
		|\Xi_{ij}|k_j
		+
		|\iota_i|
		\le
		-\epsilon_i
		<
		0
		\label{mum:CLF_derivative4},
	\end{align}
	where $\epsilon_i>0$ is a tuning parameter. The left hand side of (\ref{mum:CLF_derivative4}) can be upper bounded as
	\begin{align}
		-
		(
		1
		-
		B_{ii}
		)
		k_i
		+
		\sum_{j\not=i}
		B_{ij}
		k_j
		+
		a_i
		\le
		-\epsilon_i
		\label{mum:CLF_derivative5},
	\end{align}
	by taking into consideration
	\(
	\forall t\ge0\colon
	|\Xi_{ij}(t)|
	\le
	B_{ij}
	\)
	and
	\(
	\forall t\ge0\colon
	|\iota_i(t)|
	\le
	a_i
	\).
	Multiplying (\ref{mum:CLF_derivative5}) by
	\(
	-
	\frac{1}
	{	
		1
		-
		B_{ii}
		}
		<0
	\)
	gives
	\begin{align}
		k_i
		-
		\sum_{j\not=i}
		\frac{B_{ij}
			}
		{1
			-
			B_{ii}}
			k_j
		\ge
		\frac{
			a_i
			+
			\epsilon_i}
		{1
			-
			B_{ii}
			}
		\label{mum:CLF_derivative6},
	\end{align}
	leading to the equation system
	(\ref{mum:k_eq_sys_static})
	with solutions
	\( 
	k
	=
	\begin{bmatrix}
		k_1 & k_2 & k_3
	\end{bmatrix}^T
	\)
	and the definitions of $D$ and $z$ according to (\ref{mum:def_D_and_z}). As we assumed non-negative values for the gains $k$, it is required to show that the solution of (\ref{mum:k_eq_sys_static}) is non-negative. According to the Perron–Frobenius theorem \cite[Sec.~7.4]{slotine1991applied}, the solution 
	 $k$ of (\ref{mum:k_eq_sys_static}) has non-negative components if all entries of 
	$D$ and $z$ 
	are non-negative and the spectral radius satisfies 
	$\varrho(D)<1$, which all holds true by the made assumptions.
\end{proof}

In Theorem \ref{num:Theorem_SMC}, the bounds $B_{ij}$ represent the magnitude of uncertainties in the control effectiveness matrix, specifically capturing deviations $\Delta M_{\delta}$ of the control surfaces from their nominal behavior. The bounds $a_i$ quantify the combined effect of aerodynamic uncertainties, including both the control-surface-related ($\Delta M_{\delta}$) and the non-control-surface-related disturbances ($\Delta M_a$). Together, $B_{ij}$ and $a_i$ define the level of uncertainties that the sliding-mode switching gain 
$K\text{sgn}(s)$ must compensate. The assumption $B_{ii}<1$ guarantees that the gains $k_i$ are always non-negative, preventing control inversion along any axis. The spectral radius condition $\varrho(D)<1$ ensures that uncertainties in one axis do not dominate or reinforce those in others, avoiding unbounded gains.

The static gain SMC of Theorem \ref{num:Theorem_SMC} requires unnecessarily high gains in the nominal case, when no damage is present. This can amplify noise and increase chattering. To address this issue, we extend the results to an adaptive gain approach, formalized as the proposed RASMC in the following theorem.

\begin{theorem}[RASMC]
	\label{num:Theorem_SMC_adaptive}
	Consider the control law (\ref{num:main_control_law})
	with (\ref{num:aux_control_law_pi}) and (\ref{num:aux_control_law_mu}),
	where
	\(
	\tilde \Theta=\Theta-\Theta_r
	\) is the control tracking error,  
	\(
	\Theta_r
	=
	\begin{bmatrix}
		\phi_r & \theta_r &	\psi_r
	\end{bmatrix}^T
	\)
	are the reference attitude angles,
	$\Lambda\in\mathbb{R}^{3\times3}$
	is chosen so that matrix
	\(
	\bar \Lambda=-\Lambda
	\) is Hurwitz,	
	\(
	K=\mathrm{diag}(k)\in\mathbb{R}^{3\times3}
	\) 
	is a diagonal matrix of the adaptive controller gains
	\( 
	k
	=
	\begin{bmatrix}
		k_1 & k_2 & k_3
	\end{bmatrix}^T
	\)
	and $s$	is the sliding variable defined by (\ref{num:def_sliding_surface}).
	Let $B_{ij}\in\mathbb{R}$ be bounds of 
	\(
	\forall t\ge0\colon
	|\Xi_{ij}(t)|
	\le
	B_{ij}
	\),
	with matrix elements $[\Xi_{ij}]=\Xi\in\mathbb{R}^{3\times3}$ being defined by (\ref{mum:def_Sigma})
	and let $a_i\in\mathbb{R}$ be bounds of
	\(
	\forall t\ge0\colon
	|\iota_i(t)|
	\le
	a_i
	\),
	with elements $\iota_i$ being defined by the vector (\ref{mum:def_iota}). Define $D$ and $z$ according to (\ref{mum:def_D_and_z}), where $\epsilon_i>0$ is a tuning parameter. Assume $\varrho(D)<1$ and $\forall i \colon B_{ii}<1$
	then from
	\begin{align}
		(I_{3 \times 3}-D)	
		k_d
		=
		z
		\label{num:equ_sys_adapt_gain_bounds}
	\end{align}
	a solution 
	\(
	k_d
	=
	\begin{bmatrix}
	k_{1,d}	& k_{2,d} & k_{3,d}
	\end{bmatrix}^T
	\)
	with  non-negative components can be obtained.	
	Inspired by \cite{huang2008adaptive}, we propose the modified gain adaptation law
	\begin{align}
		\dot{k}_i =
		\begin{cases}
			(1-B_{ii})
			\frac{1}{\gamma_i}
			|s_i|
			& k_i < k_{i,d},\\
			0,
			& k_i = k_{i,d},
		\end{cases}
		\qquad
		k_i(t_0)\ge 0,
	\label{num:adaptation_law}
	\end{align}
	where $\gamma_i>0$ is a tuning parameter and $k_{i,d}$ is the upper bound of the adaptive gain obtained from the solution of (\ref{num:equ_sys_adapt_gain_bounds}).
	
	If the gains $k$ are adapted according to (\ref{num:adaptation_law}) then the tracking error converges to zero, i.\,e. $\tilde \Theta\to0$ for $t\to\infty$.
\end{theorem}
\begin{proof}
	Let $k_{i,d}$ be some static gain and 
	\(
	\tilde k_i = k_i - k_{i,d}
	\)
	be the gain error. Consider the Lyapunov function candidate
	\begin{align}
		V_2
		&=
		V_1
		+
		\sum_{i=1}^{3}
		\frac{1}{2}\gamma_i
		\tilde k_i^2.
		\label{num:def_lyapunovfct2}
	\end{align}
	with $V_1$ being defined by (\ref{num:def_LyapunovFct1}).
	Consequently,
	\(
	\dot	V_2
		=
		\dot V_1
		+
		\sum_{i=1}^{3}
		\gamma_i
		\tilde k_i\dot k_i,
	\)
	which by substitution of the adaptation law (\ref{num:adaptation_law}) yields
	\begin{align}
		\dot V_2
		&=
		\dot V_1
		+
		\sum_{i=1}^{3}
		\tilde k_i
		(1-B_{ii})
		|s_i|.
		\label{num:V2dot1}
	\end{align}
	Substituting (\ref{mum:CLF_derivative3}) in (\ref{num:V2dot1}) results in
	\begin{align}
	&
	\dot V_2
	\le
	\sum_{i=1}^{3}
	|s_i|\cdot
	\Bigg
	(
	-
	(1
	+
	\Xi_{ii}
	)
	k_i
	\notag\\
	& \qquad \qquad \qquad
	+
	\sum_{\substack{j=1 \\ j \neq i}}^{3}
	|\Xi_{ij}|k_j
	+
	|\iota_i|
	+
	(1-B_{ii})
	\tilde k_i
	\Bigg
	),
	\label{num:V2dot2}
	\end{align}
	where it is assumed that all gains $k_1, k_2, k_3$ are non-negative, which is true by applying the adaption law (\ref{num:adaptation_law}). We want $\forall s\not=0\colon\dot V_2<0$, therefore
	\begin{align}
		&
		-
		(1
		+
		\Xi_{ii}
		)
		k_i
		+
		\sum_{j\not=i}
		|\Xi_{ij}|k_j
		+
		|\iota_i|
		\notag\\
		& \qquad \qquad  \quad
		+
		(1-B_{ii})
		(k_i
		-
		k_{i,d})
		\le
		-\epsilon_i
		<
		0
		\label{num:V2dot3},
	\end{align}
	with $\epsilon_i>0$. The left hand side of (\ref{num:V2dot3}) can be upper bounded as
	\begin{align}
		&
		-
		(
		1
		-
		B_{ii}
		)
		k_i
		+
		\sum_{j\not=i}
		B_{ij}
		k_j
		+
		a_i
		\notag\\
		& \qquad \qquad \qquad \quad
		+
		(1-B_{ii})
		(
		k_i
		-
		k_{i,d}
		)
		\le
		-\epsilon_i
		\label{num:V2dot4},
	\end{align}
	by taking into consideration
	\(
	\forall t\ge0\colon
	|\Xi_{ij}(t)|
	\le
	B_{ij}
	\)
	and
	\(
	\forall t\ge0\colon
	|\iota_i(t)|
	\le
	a_i
	\). Multiplying (\ref{num:V2dot4}) with
	\(
	\frac{1}{1-B_{ii}}
	>
	0
	\) yields
	\begin{align}
		\sum_{j\not=i}
		\frac{B_{ij}}{1-B_{ii}}
		k_j
		+
		\frac{
		a_i
		+
		\epsilon_i}
		{1-B_{ii}}
		-
		k_{i,d}
		\le
		0
		\label{num:V2dot5},
	\end{align}
	We chose the static gains of as the solution of 	
	\begin{align}
		k_{i,d}
		=
		\sum_{j\not=i}
		\frac{B_{ij}
		}{
		1
		-
		B_{ii}
		}
		k_{j,d}
		+
		\frac{a_i+\epsilon_i}{1
			-
			B_{ii}}
	\label{num:V2dot6}
	\end{align}
	leading to the equation system (\ref{num:equ_sys_adapt_gain_bounds}) with definitions of $D$ and $z$ according to (\ref{mum:def_D_and_z})
	As shown in the proof of Theorem \ref{num:Theorem_SMC}, the solution $k_d$ of (\ref{num:equ_sys_adapt_gain_bounds}) has non-negative components if $\varrho(D)<1$ and $B_{ii}<1$, which all hold true by the stated assumptions. Substituting (\ref{num:V2dot6}) in (\ref{num:V2dot5}) gives
	\begin{align}
		\sum_{j\not=i}
		\frac{B_{ij}}{1-B_{ii}}
		(k_j-k_{j,d})
		\le
		0
	\label{num:V2dot8},
	\end{align}
	which is always satisfied. This follows from $\frac{B_{ij}}{1-B_{ii}}\ge 0$ and the fact that $0\le k_j\le k_{j,d}$ under the applied adaptation law (\ref{num:adaptation_law}).
	We have shown that (\ref{num:V2dot5}) and therefore also (\ref{num:V2dot3}) hold for all $k_i\le k_{i,d}$. Consequently, $V_2$ goes to zero for all $s\not=0$ and therefore $\lim_{t\to\infty}s=0$. In addition, the adaptive gains are always bounded by $k_{i,d}$ due to the adaption law (\ref{num:adaptation_law}).
\end{proof}

\section{Simulation Study}	
\label{sec:sim_study}
\begin{figure}[!tbp]
	\centering
	\includegraphics[width=1.00\linewidth]{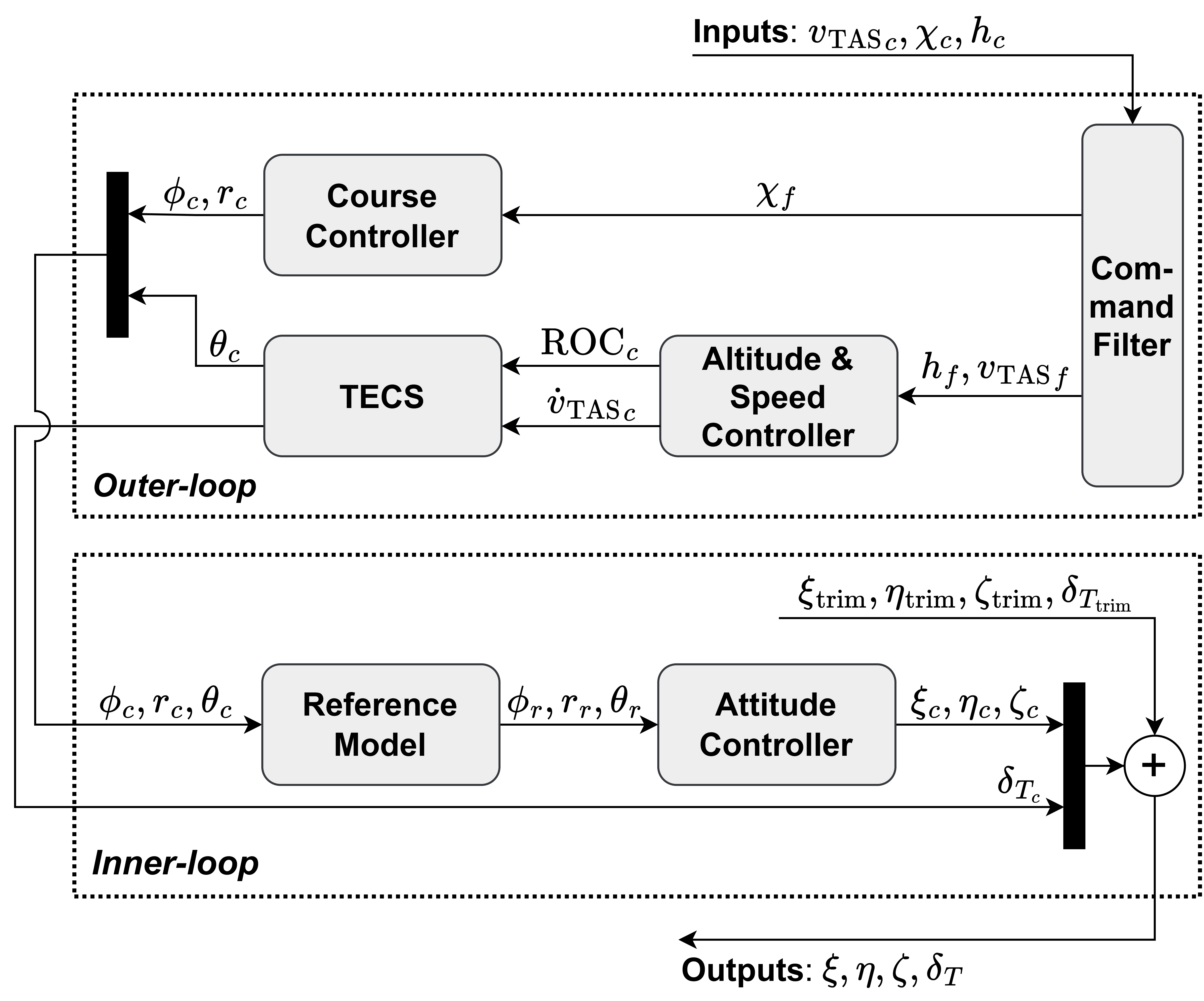}
	\caption{Autopilot with inner- and outer-loop control layers (feedback of the system states not visualized).}
	\label{fig:autopilot_structure}
\end{figure}

\begin{table}[b!]
	\centering
	\renewcommand{\arraystretch}{1.2} 
	\caption{Damage values used in simulation study}
	\label{tab:damage_parameters}
	\begin{tabular}{l c | l c | l c}
		\hline
		Parameter & Value & Parameter & Value & Parameter & Value \\
		\hline
		
		$\lambda_{l,0}$ & 0.4 & $\lambda_{l,\alpha}$ & 0.4 & $\lambda_{l,\beta}$ & 0.2 \\
		$\lambda_{l,p}$ & 0.4 & $\lambda_{l,q}$ & 0.6 & $\lambda_{l,r}$ & 0.2 \\
		$\lambda_{l,\xi}$ & 0.4 & $\lambda_{l,\eta}$ & 0.2 & $\lambda_{l,\zeta}$ & 0.3 \\
		$\lambda_{l}^{\mathrm{asym}}$ & -0.015 & $\lambda_{m,0}$ & 0.1 & $\lambda_{m,\alpha}$ & 0.3 \\
		$\lambda_{m,\beta}$ & 0.5 & $\lambda_{m,p}$ & 0.2 & $\lambda_{m,q}$ & 0.9 \\
		$\lambda_{m,r}$ & 0.1 & $\lambda_{m,\xi}$ & 0.1 & $\lambda_{m,\eta}$ & 0.6 \\
		$\lambda_{m,\zeta}$ & 0.2 & $\lambda_{m}^{\mathrm{asym}}$ & 0.03 & $\lambda_{n,0}$ & 0.8 \\
		$\lambda_{n,\alpha}$ & 0.8 & $\lambda_{n,\beta}$ & 0.9 & $\lambda_{n,p}$ & 0.8 \\
		$\lambda_{n,q}$ & 0.9 & $\lambda_{n,r}$ & 0.8 & $\lambda_{n,\xi}$ & 0.1 \\
		$\lambda_{n,\eta}$ & 0.1 & $\lambda_{n,\zeta}$ & 0.1 & $\lambda_{n}^{\mathrm{asym}}$ & -0.03 \\
		$\lambda_{L,0}$ & 0.5 & $\lambda_{L,\alpha}$ & 0.5 & $\lambda_{L,\beta}$ & 0.3 \\
		$\lambda_{L,p}$ & 0.5 & $\lambda_{L,q}$ & 0.7 & $\lambda_{L,r}$ & 0.3 \\
		$\lambda_{L,\xi}$ & 0.3 & $\lambda_{L,\eta}$ & 0.3 & $\lambda_{L,\zeta}$ & 0.3 \\
		$\lambda_{L}^{\mathrm{asym}}$ & -0.02 & $\lambda_{Y,0}$ & 0.1 & $\lambda_{Y,\alpha}$ & 0.6 \\
		$\lambda_{Y,\beta}$ & 0.1 & $\lambda_{Y,p}$ & 0.2 & $\lambda_{Y,q}$ & 0.3 \\
		$\lambda_{Y,r}$ & 0.4 & $\lambda_{Y,\xi}$ & 0.2 & $\lambda_{Y,\eta}$ & 0.3 \\
		$\lambda_{Y,\zeta}$ & 0.1 & $\lambda_{Y}^{\mathrm{asym}}$ & -0.01 & $\lambda_{D}^{\mathrm{asym}}$ & 0.1 \\
		
		\hline
	\end{tabular}
\end{table}

The damage-aware flight dynamics model introduced in Section \ref{sec:modeling_UAV} is implemented in a Python-based simulation environment and simulated in closed-loop with the proposed RASMC control law to evaluate time-domain performance. The study considers the high-speed, jet-powered fixed-wing UAV demonstrator Proteus \cite{dlr200325}, operated by DLR. Nominal aerodynamic coefficients are computed using Athena Vortex Lattice (AVL) based on the aircraft’s wing geometry \cite{DrelaAVL}.

\begin{figure*}[!tbp]
	\subfloat[Attitude tracking performance\label{fig:num_exampl_sim_results_a}]{%
		\includegraphics[width=0.49\linewidth]{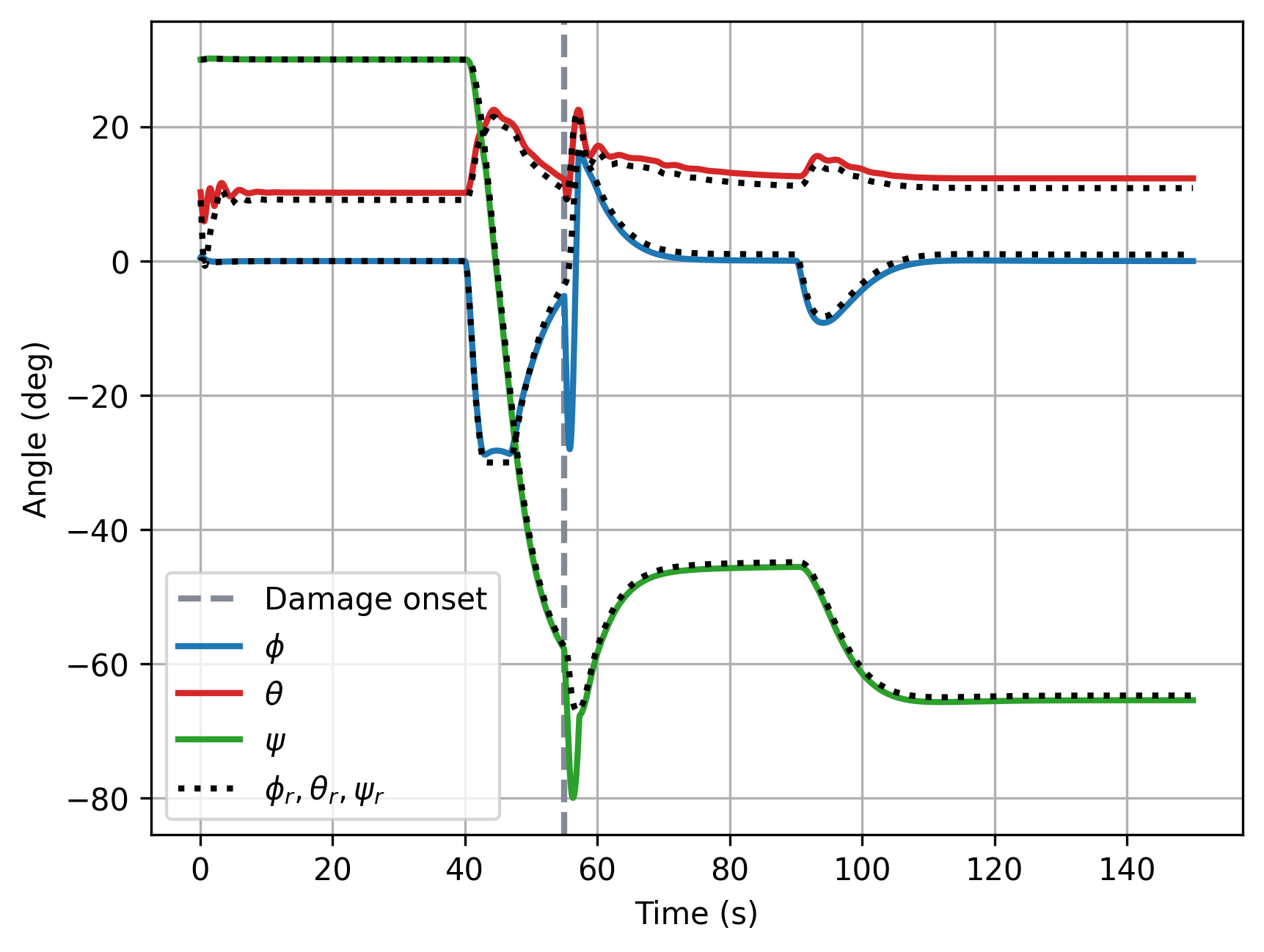}}
	\hfill
	\subfloat[Outer-loop tracking performance\label{fig:num_exampl_sim_results_b}]{\includegraphics[width=0.49\linewidth]{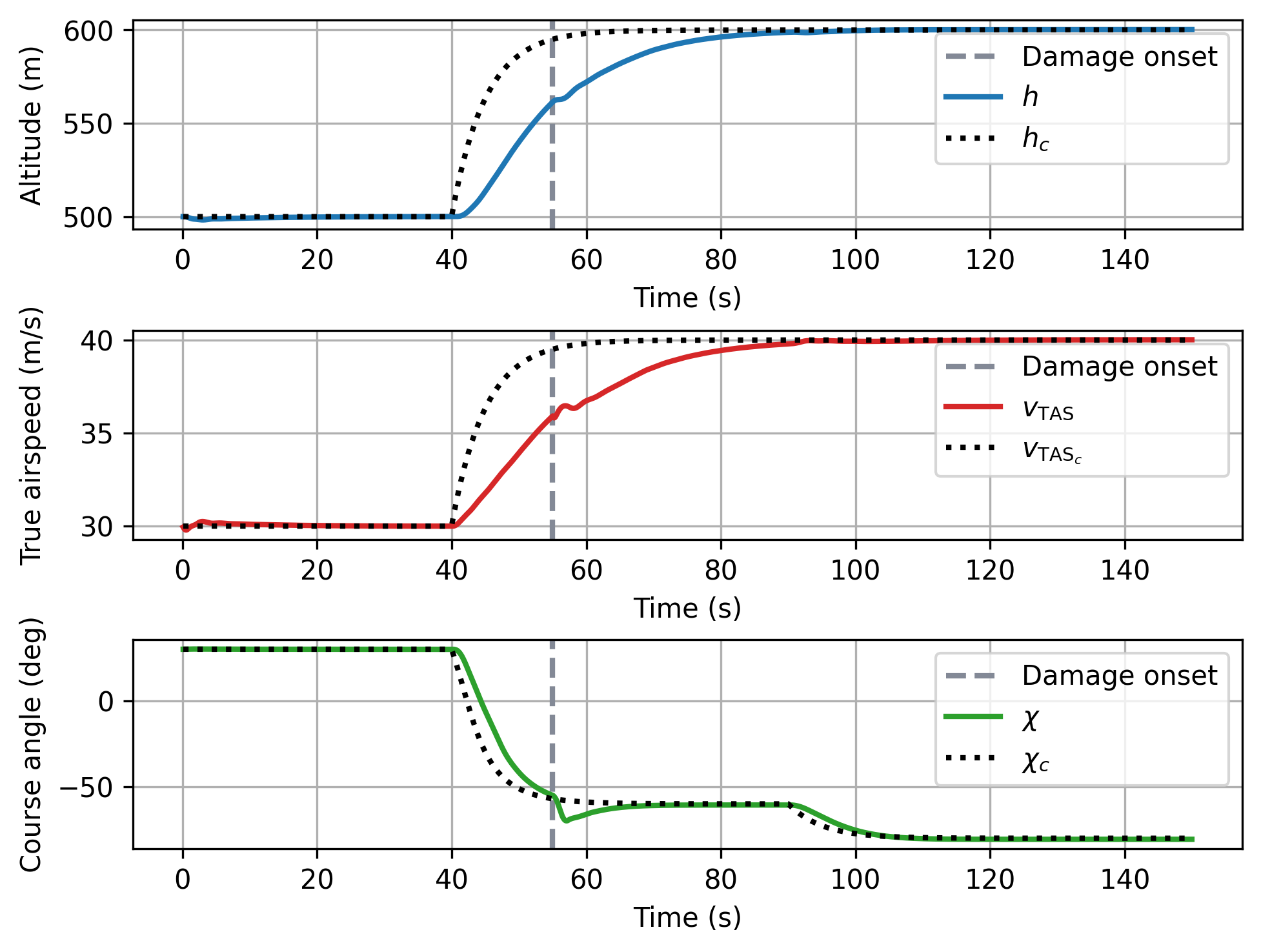}}\\
	\subfloat[Control inputs\label{fig:num_exampl_sim_results_c}]{\includegraphics[width=0.49\linewidth]{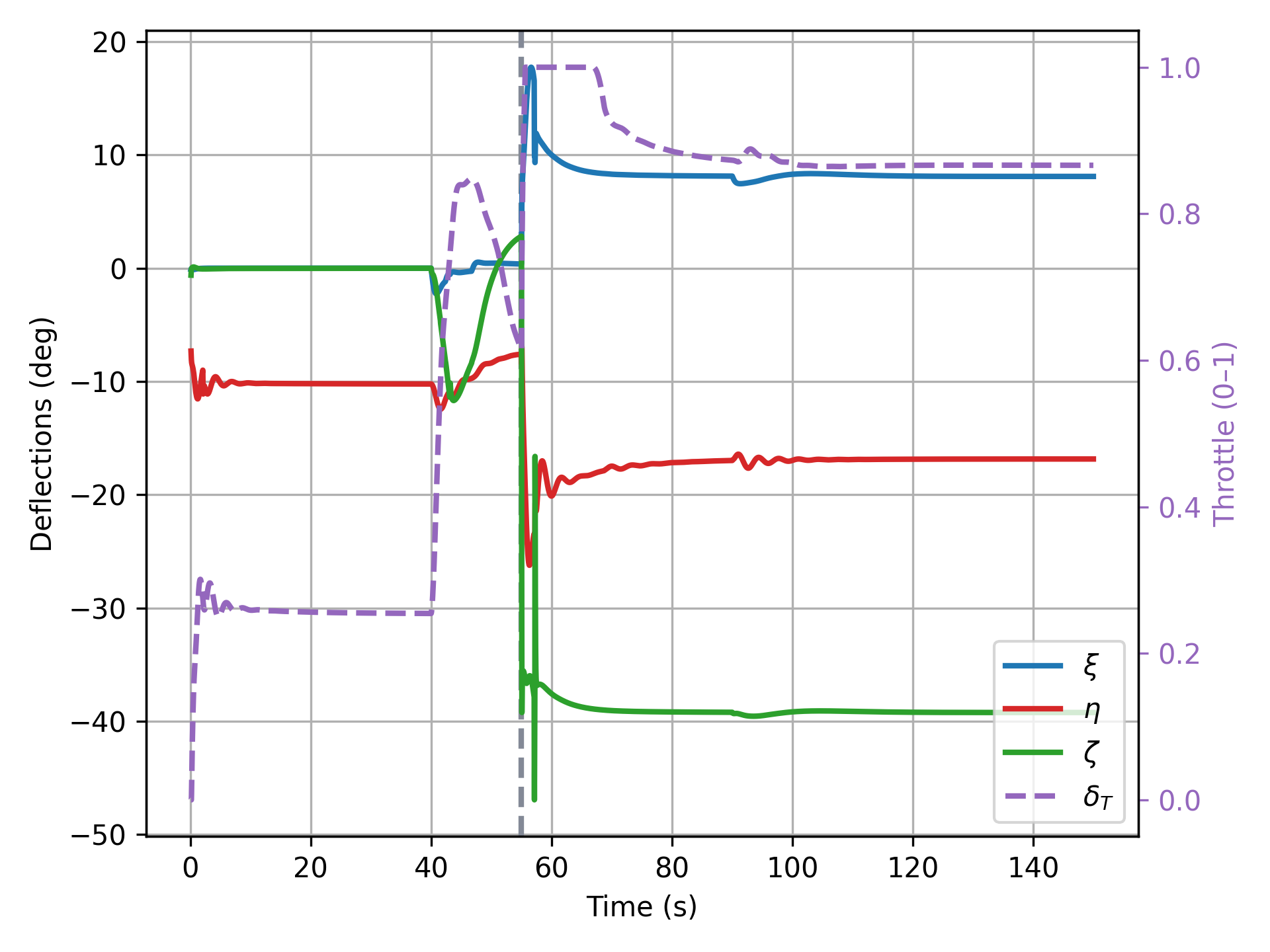}}
	\hfill
	\subfloat[Sliding variables\label{fig:num_exampl_sim_results_d}]{\includegraphics[width=0.49\linewidth]{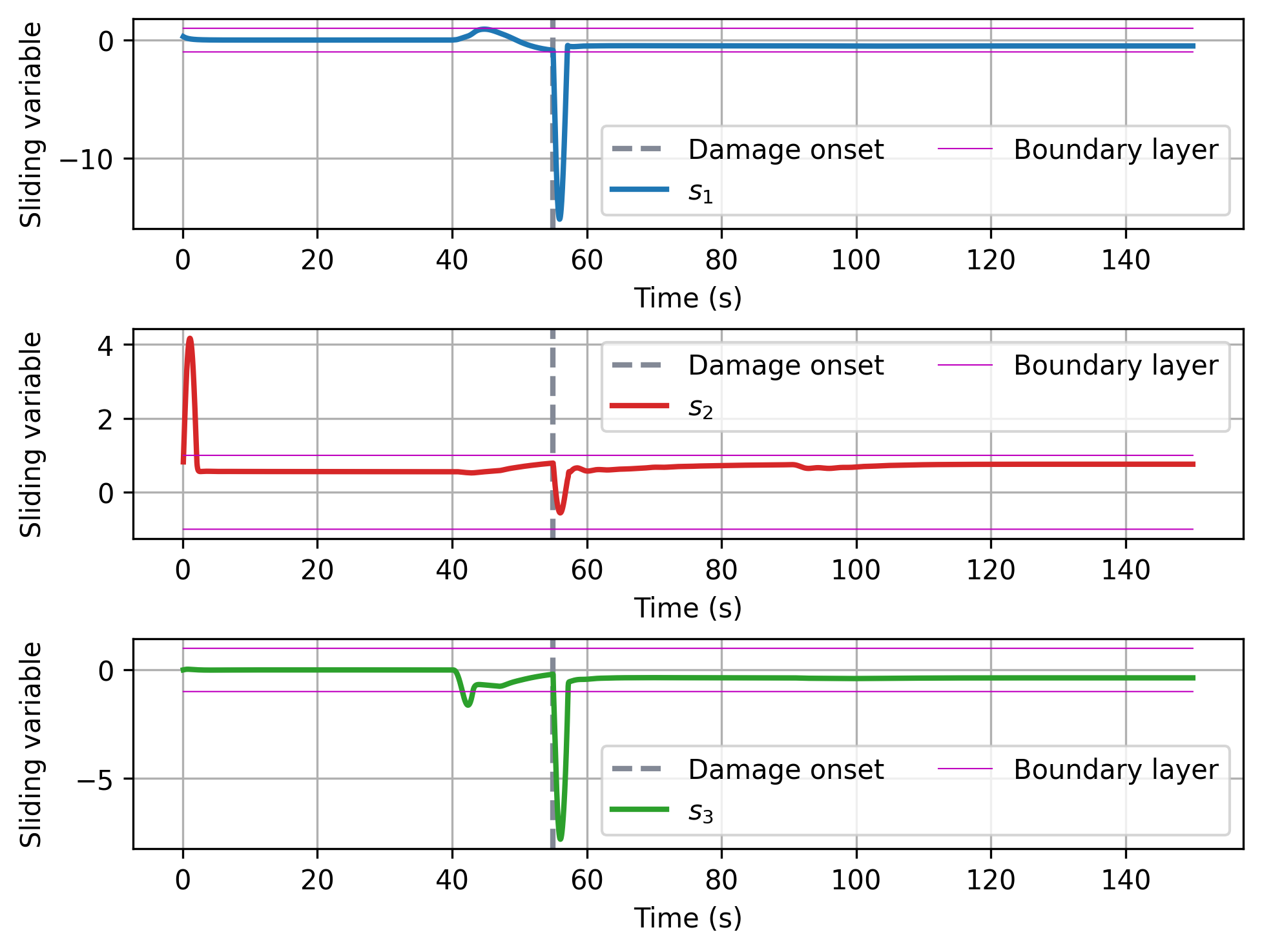}}
	\\
	\centering
	\subfloat[Adaptive gains\label{fig:num_exampl_sim_results_e}]{\includegraphics[width=0.49\linewidth]{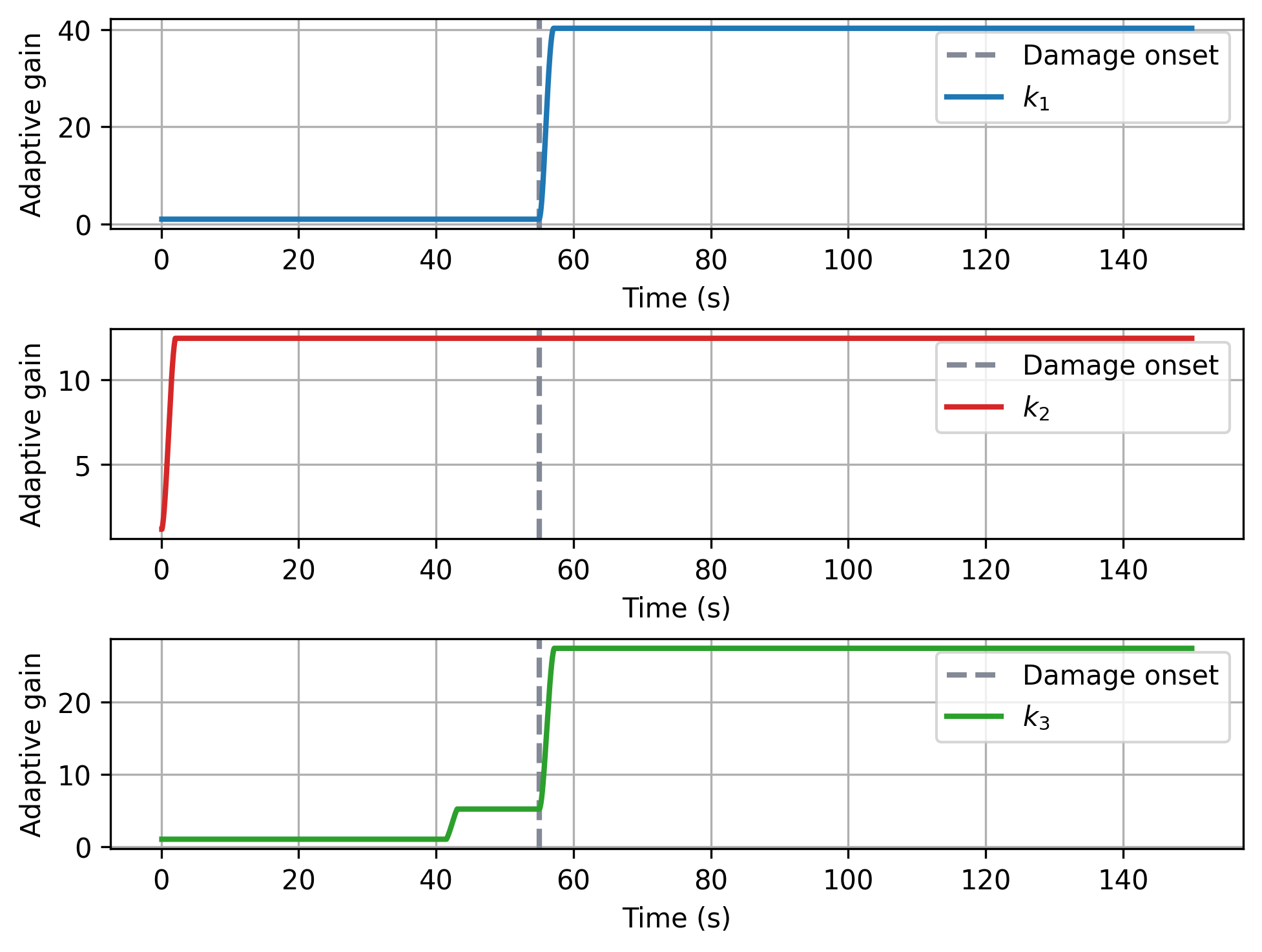}}
	\caption{Controller performance evaluation of RASMC} 
	\label{fig:num_exampl_sim_results}
\end{figure*}

The RASMC is embedded into the autopilot structure of Fig. \ref{fig:autopilot_structure}, comprising inner- and outer-loop control layers. The inputs of the autopilot are the commands for airspeed ${v_{\mathrm{VTAS}}}_c$, course $\xi_c$, and altitude $h_c$. In the outer-loop layer, these commands are processed by a command filter to smooth step changes. Based on the filtered outputs, the course controller generates commands for roll $\phi_c$ and yaw rate $r_c$. The altitude and speed controllers generate commands for rate of climb $\mathrm{ROC}_c$ and airspeed rate $\dot v_{{\mathrm{TAS}}_c}$, which are used by the total energy control system (TECS) to compute the commanded pitch $\theta_c$ and throttle $\delta_{T_c}$. In the inner-loop layer, the commanded attitude is processed through a reference model to generate the desired response behavior of the UAV along each axis, producing the reference attitude 
$\Theta_r$ and its corresponding time derivatives $\dot \Theta_r$, $\ddot \Theta_r$. The reference $\Theta_r$ serves as input for the attitude controller block, which implements the proposed robust sliding-mode controller of Section \ref{sec:rob_controller}. After adding the trim values, the autopilot outputs the deflections $\delta$ and the throttle $\delta_{T}$ which are then applied to the implemented dynamic equations (\ref{num:rigid_body_dynamics_explicit}) and a jet turbine model that generates the thrust $T$. The autopilot is run at a fixed sampling rate of \SI{50}{Hz}, while the flight dynamics are numerically integrated using the variable-step solver LSODA \cite{petzold1983automatic} implemented in SciPy. We consider no wind effects or sensor noise.

The aerodynamic damage is introduced at 
\(
t_d
=
\SI{55}{s}
\)
by modifying the aircraft's stability and control derivatives according to the damage parameters shown in Table \ref{tab:damage_parameters}. 

To achieve robust control of the damaged UAV, we employ the RASMC approach of Theorem \ref{num:Theorem_SMC_adaptive}. A smoothing boundary layer \cite{slotine1991applied} is introduced to mitigate chattering. Consequently, the switching term $\text{sgn}(s)$ is replaced by a saturation function 
\(
\text{sat}(s)
=
\begin{bmatrix}
	\text{sat}(s_1) & \text{sat}(s_2) & \text{sat}(s_3)
\end{bmatrix}^T
\)
with $\text{sat}(s_i)$ being defined as
\begin{align}
	\text{sat}(s_i) =
	\begin{cases} 
		+1 & s_i > \sigma_i \\
		\frac{s_i}{\sigma_i} & |s_i| \le \sigma_i \\
		-1 & s_i < -\sigma_i
	\end{cases}, \quad i = 1,2,3
\end{align}
and $\sigma_i>0$ denoting the boundary layer width. The introduction of the boundary layer reduces chattering, but as a trade-off, $s_i$ cannot be driven exactly to zero anymore and is instead confined to the domain $|s_i| \le \sigma_i$. However, this implies that the adaptive gains permanently increase as the adaptation (\ref{num:adaptation_law}) directly depends on $|s_i|$. Therefore, we slightly modify the adaption to
\begin{align}
	\dot{k}_i =
	\begin{cases}
		(1-B_{ii})
		\frac{1}{\gamma_i}
		|s_i|
		& k_i < k_{i,d}
		\land |s_i|>\sigma_i,\\
		0,
		& \text{else},
	\end{cases}
	\label{num:adaptation_law_mod}
\end{align}
in order to work with the boundary layer. We choose initial positive values for the gains $\forall i\in\{1,2,3\} \colon k_i(t_0)=1$.

To calculate suitable upper bounds $k_d$ for the adaptive gains according to (\ref{num:equ_sys_adapt_gain_bounds}), we chose $a_i=50$, $B_{ii}=0.05$, $B_{ij}=0.005$ sufficient high to account for the uncertainties with values obtained by trial and error. Subsequently, the boundary layer widths are tuned to minimize chattering with tuned values given by $\forall i\in\{1,2,3\} \colon \sigma_i =1$.

Finally, we define the error dynamics of the sliding surface (\ref{num:def_sliding_surface}) based on the matrix $\Lambda=30I_{3\times3}$.

The resulting performance of the proposed robust control approach is visualized in Fig. \ref{fig:num_exampl_sim_results}. In Fig. \ref{fig:num_exampl_sim_results_a} the attitude tracking performance before and after the onset of damage is illustrated. After damage occurrence, only a small degradation in tracking accuracy is observed. These deviations are primarily attributed to discretization effects and the boundary layer implementation. The boundary layer effectively mitigates chattering in the signals of the control inputs as shown by Fig. \ref{fig:num_exampl_sim_results_c}, but has the drawback of reducing control accuracy.  Fig. \ref{fig:num_exampl_sim_results_b} depicts the outer-loop control variables. Despite the significant degradation of aerodynamic effectiveness, the aircraft remains stable throughout the simulation. In particular, altitude, airspeed, and course are maintained within tight bounds, demonstrating that the higher-level flight objectives remain achievable under the considered damage scenario. At the onset of damage, clear compensatory behavior of the proposed robust controller can be seen in Fig. \ref{fig:num_exampl_sim_results_c}. Notably, increased deflections are commanded to counteract the introduced aerodynamic perturbations. In addition, thrust is elevated to compensate for the increased aerodynamic drag. The sliding variables and adaptive gains are shown in Figs. \ref{fig:num_exampl_sim_results_d} and \ref{fig:num_exampl_sim_results_e}. On the onset of damage, the sliding variables $|s_1|$ and $|s_3|$ exit the boundary layer, as the gains are insufficient to compensate for the induced effects. This leads to an adaption of the gains $k_1$ and $k_3$, causing the sliding variables to converge back to the boundary layer. For the gain $k_2$, adaptation is observed immediately after the start of the simulation, but not at the onset of damage. A plausible explanation is that the nominal model does not fully capture certain aspects of the baseline flight dynamics. Consequently, the RASMC compensates for these discrepancies from the beginning by increasing the effective control effort for pitch regulation. This initial adjustment for $k_2$ is sufficiently strong, so additional effects introduced by the damage do not necessitate further increase in the control gain, as evidenced by $s_2$ remaining within the boundary layer after the onset of damage.

Finally, the controller’s responses demonstrate that it effectively compensates for damage-induced disturbances, ensuring stability and preserving tracking performance.

\section{Conclusions}

A robust adaptive sliding mode control (RASMC) approach is proposed that is capable of handling aerodynamic coefficient perturbations as well as partial loss of control effectiveness. The adaptation mechanism compensates for damage-induced effects while maintaining low control gains under nominal conditions. Lyapunov-based stability analysis guarantees tracking error convergence and allows specifying bounds on the uncertainties under which this convergence is ensured. Simulation results demonstrate the effectiveness of the method in maintaining stable and accurate tracking despite damage. Future work considers the inclusion of actuator dynamics and actuator saturation limits to further enhance the practical applicability of the approach.


\addtolength{\textheight}{-12cm}   






\bibliographystyle{IEEEtran}
\bibliography{mybibfile}

\end{document}